\documentclass[11pt]{amsart}
\usepackage{amsmath,amssymb,pxfonts}
\usepackage[T1]{fontenc}
\newtheorem{proposition}{Proposition}[section]
\newtheorem{definition}{Definition}[section]
\newtheorem{theorem}{Theorem}[section]
\newtheorem{lemma}{Lemma}[section]
\newtheorem{remark}{Remark}[section]
\title[Asymptotically Arbitrage-free Approximation]{{\Large O}n {\Large A}symptotically {\Large A}rbitrage-free
{\Large A}pproximations of the {\Large I}mplied
{\Large V}olatility}
\author[M. Fukasawa]{{\large M}asaaki {\large F}ukasawa}
\address{Graduate School of Engineering Science, Osaka University,
560-8531, Japan}
\begin{document}
\maketitle
\begin{abstract}
Following-up Fukasawa and Gatheral (Frontiers of Mathematical Finance, 2022),
we prove that the BBF formula, the SABR
 formula, and the rough SABR formula
provide asymptotically arbitrage-free approximations of the
 implied volatility under, respectively, the local volatility model, the
 SABR model, and the rough SABR model.
\end{abstract}
\section{Introduction}
The implied volatility is one of the basic quantities in financial
practice. The option market prices are translated to the implied
volatilities to normalize in a sense their dependence on strike price
and maturity. The shape of the implied volatility surface characterizes
the marginal distributions of the underlying asset price process.
Other than a flat surface corresponding to the Black-Scholes dynamics, 
no exact formula of the surface is available and so, 
various approximation formulae have been investigated.
See \cite{Gatheral} for a practical guide for the volatility surface.

One of the most famous and in daily use of financial practice  is the
SABR formula proposed by Hagan et al.~\cite{HaganEtAl} for the SABR model.
After its original derivation by \cite{HaganEtAl} based on a formal
perturbation expansion, and a verification by \cite{BBF2} based on an
 asymptotic analysis of PDE,
Balland~\cite{Balland} derived the formula (for the so-called lognormal
case)
by an elegant no-arbitrage argument.
The no-arbitrage argument remains valid for non-Markovian models,
and Fukasawa and Gatheral~\cite{FG} derived an extension of the SABR
formula to a rough SABR model, where the volatility process is non-Markovian.

It has been known that the implied volatility surface of an equity
option market typically exhibits a power-law type term structure. Since
classical local-stochastic volatility models including the SABR model
are not consistent to such a term structure, 
the so-called rough volatility model has recently attracted attention,
which is the only class of continuous price models that are consistent
to the power-law;  see Fukasawa~\cite{F17, F21}.
The rough SABR model of \cite{FG} (see also \cite{F18, Musiela, FHT}) 
is a rough volatility model and the rough SABR formula derived in
\cite{FG} explicitly exhibits a power-law term structure.

The aim of this paper is to verify the SABR and rough SABR formulae
by the no arbitrage argument beyond the lognormal case.

\newpage
\section{Asymptotically Arbitrage-free Approximation} 
Here we recall the definition of 
Asymptotically Arbitrage-free Approximation (AAA for short) given by
\cite{FG} and
derive an alternative expression of the defining equation.

Let $S = \{S_t\}$ be the underlying asset price process of an option
market, and $C=\{C_t\}$ be a call option price process with
strike price $K > 0$ and maturity $T>0$. 
We regard $(S,C)$ as a 2 dimensional stochastic process defined on a
filtered probability space $(\Omega,\mathcal{F},\mathsf{P},
\{\mathcal{F}_t\}_{t\in [0,T]})$.
We assume interest and dividend rates to be zero for brevity.
Denote by
$\Sigma^{\mathrm{BS}} = \{\Sigma^{\mathrm{BS}}_t\}$ and
$\Sigma^{\mathrm{B}} = \{\Sigma^{\mathrm{B}}_t\}$ respectively
 the Black-Scholes and the Bachelier implied volatility processes
defined by
$$C_t = P^{\mathrm{BS}}(S_t,K,T-t,\Sigma^{\mathrm{BS}}_t) =
P^{\mathrm{B}}(S_t,K,T-t,\Sigma^{\mathrm{B}}_t), $$
  where $P^{\mathrm{BS}}(S,K,\tau,\sigma)$
and
$ P^{\mathrm{B}}(S,K,\tau,\sigma)$ are respectively the
Black-Scholes and Bachelier call prices with volatility $\sigma$, namely
\begin{equation*}
\begin{split}
& P^{\mathrm{BS}}(S,K,\tau,\sigma) = S\Phi(d_+) - K\Phi(d_-), \ \ d_\pm
 = 
\frac{\log \frac{S}{K}}{\sigma\sqrt{\tau}} \pm \frac{\sigma
 \sqrt{\tau}}{2},\\
& P^{\mathrm{B}}(S,K,\tau,\sigma) = \sigma\sqrt{\tau}\phi\left(
\frac{S-K}{\sigma\sqrt{\tau}} \right) +
 (S-K)\Phi\left(\frac{S-K}{\sigma\sqrt{\tau}}\right).
\end{split}
\end{equation*}
As is well-known, a sufficient and almost necessary condition for 
 $(S,C)$ to be arbitrage-free is that there exists an equivalent
probability measure
$\mathsf{Q}$ to $\mathsf{P}$ such that 
$(S,C)$ is a local martingale under $\mathsf{Q}$.
Since 
$\Sigma^{\mathrm{BS}}$ and $\Sigma^{\mathrm{B}}$ are nonlinear
transformations of the price $C$, they are not local martingales but
should satisfy certain constraints on their finite variation components under
$\mathsf{Q}$.
Assuming that $(S,C)$ is a Brownian local martingale under $\mathsf{Q}$ and
denoting by $D^{\mathrm{BS}}\mathrm{d}t$ and $D^\mathrm{B}\mathrm{d}t$
respectively the drift components of $\mathrm{d}\Sigma^{\mathrm{BS}}$
and $\mathrm{d}\Sigma^{\mathrm{B}}$, 
the constraints (the drift conditions) are 
\begin{equation*}
\begin{split}
 &  \frac{\mathrm{d}}{\mathrm{d}t}
 \langle \log S \rangle
+ 2 k
\frac{\mathrm{d}}{\mathrm{d}t}
\langle \log S, \log \Sigma^\mathrm{BS} \rangle
+ k^2
\frac{\mathrm{d}}{\mathrm{d}t}
\langle \log \Sigma^\mathrm{BS} \rangle
- (\Sigma^\mathrm{BS})^2 + 2\Sigma^\mathrm{BS}\tau D^\mathrm{BS}
\\ &= - (\Sigma^\mathrm{BS})^2\tau
\frac{\mathrm{d}}{\mathrm{d}t}
\langle \log S, \log \Sigma^\mathrm{BS} \rangle
+ \frac{(\Sigma^\mathrm{BS})^4\tau^2}{4}\frac{\mathrm{d}}{\mathrm{d}t}
\langle \log \Sigma^\mathrm{BS} \rangle
\end{split}
\end{equation*}
and
\begin{equation*}
   \frac{\mathrm{d}}{\mathrm{d}t}
 \langle \log S \rangle
+ 2 x
\frac{\mathrm{d}}{\mathrm{d}t}
\langle \log S, \log \Sigma^\mathrm{B} \rangle
+ x^2
\frac{\mathrm{d}}{\mathrm{d}t}
\langle \log \Sigma^\mathrm{B} \rangle
- (\Sigma^\mathrm{B})^2 + 2\Sigma^\mathrm{B}\tau D^\mathrm{B} = 0,
\end{equation*}
where $k= \log K/S$, $x = K-S$ and $\tau = T-t$.
Based on this observation, Fukasawa and Gatheral~\cite{FG} introduced the following notion:
\begin{definition}
Let $\hat{\Sigma} = \{\hat\Sigma_t\}$  be a positive continuous It\^o
 process on $[0,T)$
and denote by  $\hat{D}\,\mathrm{d}t$ the drift part of
 $\mathrm{d}\hat\Sigma$ under $\mathsf{Q}$.
We say $\hat\Sigma$ is an AAA of
 $\Sigma^\mathrm{BS}$ if there exist
a continuous function
 $\varphi$ on $\mathbb{R}$
and a continuous process $\Psi = \{\Psi_t\}$ on $[0,T]$ such that
\begin{equation}\label{def1}
\begin{split}
 & \Biggl|  \frac{\mathrm{d}}{\mathrm{d}t}
 \langle \log S \rangle
+ 2 k
\frac{\mathrm{d}}{\mathrm{d}t}
\langle \log S, \log \hat\Sigma \rangle
+ k^2
\frac{\mathrm{d}}{\mathrm{d}t}
\langle \log \hat\Sigma \rangle
\\ &- \hat\Sigma^2 + 2\hat\Sigma\tau \hat{D}
+\hat\Sigma^2\tau
\frac{\mathrm{d}}{\mathrm{d}t}
\langle \log S, \log \hat\Sigma \rangle
- \frac{\hat\Sigma^4\tau^2}{4}\frac{\mathrm{d}}{\mathrm{d}t}
\langle \log \hat\Sigma \rangle \Biggr| \leq \varphi(\Psi
 \hat{\Sigma})\cdot
 o_p(1)
\end{split}
\end{equation}
as $\tau = T-t \to 0$.
We say $\hat{\Sigma}$
 is an AAA of
 $\Sigma^\mathrm{B}$ if there exist
a continuous function
 $\varphi$ on $\mathbb{R}$
and a continuous process $\Psi = \{\Psi_t\}$ on $[0,T]$ such that
\begin{equation}\label{def2}
 \Biggl|  \frac{\mathrm{d}}{\mathrm{d}t}
 \langle S \rangle
+ 2 x
\frac{\mathrm{d}}{\mathrm{d}t}
\langle  S, \log \hat\Sigma \rangle
+ x^2
\frac{\mathrm{d}}{\mathrm{d}t}
\langle \log \hat\Sigma \rangle - \hat\Sigma^2 + 2\hat\Sigma\tau \hat{D}
 \Biggr| \leq \varphi(\Psi \hat{\Sigma})\cdot
 o_p(1)
\end{equation}
as $\tau = T-t \to 0$.
\end{definition}

\begin{remark}\upshape
In \cite{FG}, as usual, $o_p(1)$ is interpreted as a term which converges to
$0$ in probability. This is however not really reflecting the idea behind
AAA in that a constant term $a = a\hat\Sigma/\hat\Sigma$
 also has a form of $\varphi(\Psi\hat\Sigma)
\cdot o_p(1)$ when $\hat\Sigma \to \infty$ in
probability. We indeed have  such a divergence for any fixed $K>0$ when
 considering the rough SABR approximation. To remedy this,
in the present article, we interpret $o_p(1)$ as a term
which converges to $0$ uniformly in the strike price $K > 0$.
\end{remark}

\begin{remark}\upshape
Any locally bounded function on $\mathbb{R}$ is dominated by a
continuous function. Therefore the continuity property of $\varphi$ in
the definition of AAA can be replaced by the local boundedness of $\varphi$.
\end{remark}

The first observation of this study is the following simplification by
It\^o's formula:
for a continuous semimartingale $X$ and 
a positive continuous semimartingale $\hat\Sigma$,
\begin{equation*} 
\mathrm{d}
 \langle X \rangle
- 2 X
\mathrm{d}
\langle X, \log \hat\Sigma \rangle
+ X^2
\mathrm{d}
\langle \log \hat\Sigma \rangle
= \hat\Sigma^2 \mathrm{d}\left\langle \frac{X}{\hat\Sigma} \right\rangle.
\end{equation*}
Recalling that $k =\log K/S$ and $x = K-S$ and so that
$\mathrm{d}k = -\log S$ and $\mathrm{d}x = - \mathrm{d}S$, we conclude
that \eqref{def1} and \eqref{def2} are respectively equivalent to 
\begin{equation}\label{def3}
\Biggl|
\frac{\mathrm{d}}{\mathrm{d}t}
 \left\langle \frac{k}{\hat\Sigma} \right\rangle
- 1 + \frac{2\tau \hat{D}}{\hat\Sigma}
-\tau
\frac{\mathrm{d}}{\mathrm{d}t}
\langle k, \log \hat\Sigma \rangle
- \frac{\tau^2}{4}\frac{\mathrm{d}}{\mathrm{d}t}
\langle \hat\Sigma \rangle \Biggr| \leq \frac{\varphi(\Psi
 \hat{\Sigma})}{\hat\Sigma^2}\cdot
 o_p(1)
\end{equation}
and
\begin{equation*}\label{def4}
\Biggl|
\frac{\mathrm{d}}{\mathrm{d}t}
 \left\langle \frac{x}{\hat\Sigma} \right\rangle
- 1 + \frac{2\tau \hat{D}}{\hat\Sigma}
 \Biggr| \leq \frac{\varphi(\Psi
 \hat{\Sigma})}{\hat\Sigma^2}\cdot
 o_p(1).
\end{equation*}

We conclude this section with some additional definitions.
\begin{definition}
We say 
$f$ is a $C^{1a}$ function  if 
 $f$ is differentiable,
the derivative $f^\prime$ is absolutely continuous,
and the Radon-Nikodym derivative $f^{\prime\prime}$ of $f^\prime$ is
locally bounded on the domain of $f$.
We say $g$ is a $C^{2+}$ function if  $g$ is $C^2$ and
\begin{equation*}
 f(x) = 
\begin{cases}
 \frac{x}{g(x)} & x \neq 0, \\
\frac{1}{g^\prime(0)}  & x = 0
\end{cases}
\end{equation*}
is a positive $C^{1a}$ function on $\mathbb{R}$. 
\end{definition}
\begin{remark}
 \upshape
By the Tanaka formula, It\^o's formula for $C^2$ functions 
remains true for $C^{1a}$ functions (see e.g., Section~1, Chapter~VI of
\cite{RY}).
\end{remark}
\begin{remark}
 \upshape
By Lemma~\ref{ratio}, if $g$ is $C^3$
with
$g(0) = 0$ and
$g^\prime > 0$, then $g$ is a $C^{2+}$ function.
\end{remark}

\section{The BBF and SABR formulae}

The BBF formula refers to an approximation formula
\begin{equation*}\label{BBF}
\hat\Sigma = \frac{\log
 \frac{K}{S}}{\int_{S}^K\frac{\mathrm{d}s}{v(s,T)}}
\end{equation*}
under a local volatility model
\begin{equation}\label{local}
 \mathrm{d}S_t = v(S_t,t)\mathrm{d}W_t,
\end{equation}
where $W$ is a standard Brownian motion. We assume that $v$ is a positive $C^{2,0}$
function on $(0,\infty)\times [0,T]$ and
$S$ is positive and
continuous on $[0,T]$. Here, we interpret
\begin{equation*}
 \hat\Sigma_t = \lim_{S \to K} 
\frac{\log
 \frac{K}{S}}{\int_{S}^K\frac{\mathrm{d}s}{v(s,T)}} = \frac{v(K,T)}{K}
\end{equation*}
when $S_t = K$.
By Berestycki et al.~\cite{BBF1}, we know that
$\lim_{t \to T}\Sigma^\mathrm{BS}_t = \hat\Sigma_T$.
In contrast to a technical argument in \cite{BBF1}, 
here we can 
easily verify the BBF formula in the sense that it provides an AAA.

\begin{proposition}
 
Let $f$ be a positive 
$C^{1a}$ function on $(0,\infty)$ with $f(K) \neq 0$, and let
\begin{equation*}
 \hat{\Sigma} = f(S).
\end{equation*}
Then, $\hat\Sigma$ is an AAA of $\Sigma^\mathrm{BS}$
 under \eqref{local} if and only if
\begin{equation*}
 f(s) =  \frac{\log
 \frac{K}{s}}{\int_{s}^K\frac{\mathrm{d}x}{v(x,T)}}
\end{equation*}
for $s \neq K$ and $f(K) = v(K,T)/K$.
\end{proposition}
\begin{proof}
Let
\begin{equation*}
 g(s) = \frac{\log \frac{K}{s}}{f(s)}.
\end{equation*}
 Then, $g$ is a $C^{1a}$ function and
%\begin{equation*}
$ \frac{k}{\hat\Sigma} = g(S) $.
%\end{equation*}
By the It\^o-Tanaka formula,
\begin{equation*}
\mathrm{d} \left\langle  \frac{k}{\hat\Sigma}\right\rangle
= g^\prime(S)^2\mathrm{d}\langle S \rangle = 
g^\prime(S)^2v(S,t)^2\mathrm{d}t
\end{equation*}
and
$\hat{D}$, $\hat\Sigma^{-1}$, and $ \frac{\mathrm{d}}{\mathrm{d}t}\langle
 \hat\Sigma \rangle$ are $O_p(1)$ as $t \to T$.
Therefore \eqref{def3} is satisfied if and only if
\begin{equation*}
| g^\prime(s)v(s,T) | = 1.
\end{equation*}
Since $g(0) = 0$ and $f$ is positive, this is equivalent to
\begin{equation*}
 g(s) = \int_s^K \frac{\mathrm{d}x}{v(s,T)}.
\end{equation*}
\end{proof}
The same argument with $k$ replaced by $x = K-S$ shows that 
\begin{equation*}
 \hat\Sigma = \frac{K-S}{\int_{S}^K \frac{\mathrm{d}s}{v(s,T)}}
\end{equation*}
is an AAA of $\Sigma^\mathrm{B}$ under \eqref{local}.\\

The SABR formula refers to 
\begin{equation}\label{SABRf}
 \hat\Sigma = \frac{\nu k}{g(Y)}, \ \ Y
= \frac{\nu}{\alpha}\int_S^K\frac{\mathrm{d}s}{\beta(s)}, \ \ 
g(y) = - \log \frac{\sqrt{1+ 2\rho y + y^2}-y-\rho}{1-\rho}
\end{equation}
proposed by Hagan et al.~\cite{HaganEtAl} as
an approximation of $\Sigma^\mathrm{BS}$ under
the SABR model
\begin{equation}\label{SABRm}
 \mathrm{d}S_t = \alpha_t \beta(S_t) \mathrm{d}Z_t, \ \ 
\mathrm{d}\alpha_t = \nu \alpha_t \mathrm{d}W_t, \ \ 
\mathrm{d}\langle Z,W\rangle_t = \rho \mathrm{d}t
\end{equation}
with $ \mathrm{d}\langle Z\rangle_t 
= \mathrm{d}\langle W\rangle_t  
= \mathrm{d}t$,
where $\rho \in (-1,1)$ and $\nu > 0$ are constants,
and $\beta$ is a positive $C^2$
function on $(0,\infty)$.
Here we assume that $S$ is positive and continuous on $[0,T]$.
The validity of the approximation has been discussed by Berestycki et
al.~\cite{BBF2}, Osajima~\cite{Osajima}  and others.
In the lognormal case, that is, $\beta(s)= s$, 
Balland~\cite{Balland} shows that the approximation \eqref{SABRf} is, in our
terminology, an AAA.
Our alternative expression \eqref{def3} allows us to observe it is the case
in general.
\begin{proposition}
 
Let $g$ be a $C^{2+}$ function on $\mathbb{R}$ and 
\begin{equation*}
 \hat{\Sigma} =
  \frac{\nu k}{g(Y)}
, \ \ Y = \frac{\nu}{\alpha}\int_S^K\frac{\mathrm{d}s}{\beta(s)}.
\end{equation*}
Then, $\hat\Sigma$ is  an AAA of $\Sigma^\mathrm{BS}$
under \eqref{SABRm} if and only if
\begin{equation}\label{SABRg}
 g(y) = - \log \frac{\sqrt{1+ 2\rho y + y^2}-y-\rho}{1-\rho}.
\end{equation}
\end{proposition} 
\begin{proof}
 Since
\begin{equation*}
 \frac{k}{\hat\Sigma} = \frac{g(Y)}{\nu}
\end{equation*}
and
\begin{equation*}
 \mathrm{d}Y = -\frac{\nu}{\alpha \beta(S)} \mathrm{d}S
- \frac{Y}{\alpha}\mathrm{d}\alpha + \text{drift}
= -\nu \mathrm{d}Z - \nu Y \mathrm{d}W +  \text{drift},
\end{equation*}
we have
\begin{equation*}
 \mathrm{d}\left\langle \frac{k}{\hat\Sigma} \right\rangle
= \frac{g^\prime(Y)^2}{\nu^2}\mathrm{d}\langle Y \rangle
= g^\prime(Y)^2(1 + 2 \rho Y + Y^2)\mathrm{d}t.
\end{equation*}
By the It\^o-Tanaka formula, 
 $\hat{D}$, $\hat\Sigma^{-1}$, and $\frac{\mathrm{d}}{\mathrm{d}t}\langle
 \hat\Sigma \rangle$ are $O_p(1)$ as $t \to T$.
Therefore, \eqref{def3} is satisfied if and only if $g$ solves
$$ g^\prime(y)^2(1 + 2 \rho y + y^2) = 1.$$
The unique solution of this ordinary differential equation
with $g(0) = 0$ and $g(y)/y >0$ for $y\neq 0$ is the one given by \eqref{SABRg}.
\end{proof}
The same argument with $k$ replaced by $x = K-S$ shows that 
\begin{equation*}
 \hat\Sigma = \frac{\nu x}{g(Y)}
\end{equation*}
is an AAA of $\Sigma^\mathrm{B}$ under \eqref{SABRm}.

\section{The rough SABR formula}
Here we consider 
the rough SABR formula proposed by Fukasawa and Gatheral~\cite{FG}.
Suppose a positive continuous local martingale $S$ follows
\begin{equation}\label{rSABRm}
 \mathrm{d}S_t = \alpha_t \beta(S_t)\mathrm{d}Z_t, \ \ 
\mathrm{d}\xi_t(s) = \zeta(s-t)\xi_t(s) \mathrm{d}W_t,\ \ 
\mathrm{d}\langle Z, W \rangle_t = \rho \mathrm{d}t
\end{equation}
with $ \mathrm{d}\langle Z\rangle_t 
= \mathrm{d}\langle W\rangle_t  
= \mathrm{d}t$, 
where $\alpha_t = \sqrt{\xi_t(t)}$,
$\zeta(t) = \eta\sqrt{2H}(t_+)^{H-1/2}$,
$\beta$ is a positive $C^2$ function on $(0,\infty)$ and $H \in (0,1/2]$,
$\rho \in (-1,1)$ and $\eta > 0$ are constants.
The approximation formula to be examined is
\begin{equation}\label{rSABRf}
 \hat\Sigma = \frac{\zeta(\tau)k}{g(Y)}, \ \ Y =
  \frac{\zeta(\tau)}{U}\int_{S}^K \frac{\mathrm{d}s}{\beta(s)}, \ \
 U = \sqrt{\frac{1}{\tau}\int_t^{T} \xi_t(s)\mathrm{d}s},
\end{equation}
where $g$ is a solution of
 the differential equation
\begin{equation}\label{ode}
 g^\prime(y)^2\left(1 + 2\rho \frac{y}{2H+1} +
	       \frac{y^2}{(2H+1)^2}\right)
= 1 - (1-2H)\left(1 - \frac{yg^\prime(y)}{g(y)}\right)
\end{equation}
with
\begin{equation}\label{gp}
g(0) = 0, \ \ g^\prime(0) > 0. 
\end{equation}
See Lemma~\ref{exist} below for the unique existence of the solution.
Further by Lemmas~\ref{ratio} and \ref{exist}, the solution $g$ is a $C^{2+}$
function.
In the lognormal case, that is, $\beta(s) = s$, it is shown in
\cite{FG} that the formula \eqref{rSABRf}
gives an AAA of $\Sigma^\mathrm{BS}$.
It is however left open in \cite{FG}
whether \eqref{rSABRf} is an AAA for general
$\beta$.
\begin{theorem}
Let $g$ be a $C^{2+}$ function and let $\hat\Sigma$ be defined as in
 \eqref{rSABRf}. Then, $\hat{\Sigma}$ is an AAA of
 $\Sigma^\mathrm{BS}$ under \eqref{rSABRm} if and only if 
$g$ solves \eqref{ode} with \eqref{gp}.
\end{theorem}
\begin{proof}
 Since
\begin{equation*}
 \frac{k}{\hat\Sigma} = \frac{g(Y)}{\zeta(\tau)}
\end{equation*}
and
\begin{equation}\label{dy}
\begin{split}
 \mathrm{d}Y = &- \frac{\zeta(\tau)}{U\beta(S)}\mathrm{d}S
- \frac{Y}{U}\mathrm{d}U 
- \frac{\zeta^\prime(\tau)}{\zeta(\tau)} Y \mathrm{d}t
\\ & +
\frac{1}{2}\frac{\zeta(\tau)\beta^\prime(S)}{U\beta(S)^2}\mathrm{d}\langle
S \rangle +  \frac{\zeta(\tau)}{U^2 \beta(S)}\mathrm{d}\langle S, U \rangle
+ \frac{Y}{U^2}\mathrm{d}\langle U \rangle,
\end{split}
\end{equation}
we have
\begin{equation*}
 \mathrm{d}\left\langle
\frac{k}{\hat\Sigma} \right\rangle
=\frac{g^\prime(Y)^2}{\zeta(\tau)^2}\mathrm{d}\langle Y \rangle = 
 g^{\prime}(Y)^2\left(
\frac{\alpha^2}{U^2}\mathrm{d}t +
\frac{2\alpha Y }{U\zeta(\tau)}\mathrm{d}\langle Z,\log U \rangle +
 \frac{Y^2}{\zeta(\tau)^2}\mathrm{d}\langle \log U \rangle
\right).
\end{equation*}
As is observed in \cite{FG},
\begin{equation}\label{du}
\frac{\mathrm{d}U}{U} =
 \frac{1}{2}\zeta(\tau)\,R\,\mathrm{d}W + 
\frac{1}{2}\left(
\frac{1}{\tau}\left(1 -
	           \frac{\alpha^2}{U^2}\right)-\frac{1}{4}\zeta(\tau)^2
R^2
\right) \mathrm{d}t,
\end{equation}
where
\begin{equation*}
 R_t = \frac{\int_t^T\zeta(s-t)\xi_t(s)
  \mathrm{d}s}{\zeta(\tau)\int_t^T\, \xi_t(s) \mathrm{d}s}.
\end{equation*}
Therefore,
\begin{equation*}
 \frac{\mathrm{d}}{\mathrm{d}t} \langle Z, \log U \rangle
= \frac{1}{2}\zeta(\tau) R \rho, \ \ 
 \frac{\mathrm{d}}{\mathrm{d}t} \langle \log U \rangle
= \frac{1}{4}\zeta(\tau)^2R^2
\end{equation*}
and so
\begin{equation}\label{left}
 \frac{\mathrm{d}}{\mathrm{d}t}\left\langle
\frac{k}{\hat\Sigma} \right\rangle
=
 g^{\prime}(Y)^2\left(
\frac{\alpha^2}{U^2}+
\frac{\alpha }{U} R \rho Y + 
 \frac{R^2}{4}Y^2
\right).
\end{equation}
On the other hand, recalling that $f(y): = y/g(y)$ is a positive 
$C^{1a}$ function,
\begin{equation*}
 \hat\Sigma = U f(Y)\hat\Sigma^0, \ \ \hat\Sigma^0 =  \frac{k}{\int_S^K\frac{\mathrm{d}s}{\beta(s)}}
\end{equation*}
and so
\begin{equation*}
\begin{split}
\frac{\mathrm{d}\hat\Sigma}{\hat\Sigma} 
&=  \mathrm{d}\log\hat\Sigma + \frac{1}{2}
\mathrm{d}\langle \log \hat\Sigma \rangle\\
&= \mathrm{d} \log U + 
\mathrm{d}\log f(Y) + \mathrm{d} \log \hat\Sigma^0 +  \dots
\end{split}
\end{equation*}
Therefore, 
denoting by $\hat{D}\mathrm{d}t$ the drift part of
 $\mathrm{d}\hat\Sigma$ and
using \eqref{dy} and \eqref{du}, we have
\begin{equation}\label{right}
\begin{split}
 \frac{2\tau \hat{D}}{\hat\Sigma} &=
(1-2H)\left( 1 - \frac{\alpha^2}{U^2} +
 \frac{Yf^\prime(Y)}{f(Y)}\frac{\alpha^2}{U^2} \right) +\varphi(Y) \cdot
 O_p(\tau^{2H})
\\ &=
\left(1-2H\right)\left(1 - \frac{Yg^\prime(Y)}{g(Y)}\right)
+ \frac{Yg^\prime(Y)}{g(Y)} \left(1-\frac{\alpha^2}{U^2}\right)
+ \varphi(Y) \cdot O_p(\tau^{2H}) 
\end{split}
\end{equation}
for some locally bounded function $\varphi$.
By the continuity of $f$, $\varphi(Y)$ is dominated by
 $\tilde{\varphi}(\Psi \hat\Sigma)$ with 
a continuous function $\tilde \varphi$ and a continuous process $\Psi$.
By Lemma~A.1 of \cite{FG}, we know 
$\alpha/U \to 1$ and $R \to 1/(H+1/2)$.
Then \eqref{left} and \eqref{right} imply that 
\eqref{def3} is satisfied if and only if \eqref{ode} holds.
\end{proof}
The same argument with $k$ replaced by $x = K-S$ shows that 
\begin{equation*}
 \hat\Sigma = \frac{\zeta(\tau)x}{g(Y)}
\end{equation*}
is an AAA of $\Sigma^\mathrm{B}$ under \eqref{rSABRm}.\\

\appendix
\section{Lemmas}
\begin{lemma}\label{ratio}
 Let $g$ be a $C^2$ function on $\mathbb{R}$ with
\begin{equation}\label{exp0}
\begin{split}
& g(y) = g^\prime(0)y + g^{\prime\prime}(0)\frac{y^2}{2} + O(y^3), \\ 
&g^\prime(y) = g^\prime(0)+g^{\prime\prime}(0)y + O(y^2), \\
&g^{\prime\prime}(y) = g^{\prime\prime}(0) + O(y)
\end{split}
\end{equation}
as $y \to 0$. Then, the function $G$ defined by
\begin{equation*}
 G(y) =
\begin{cases}
g(y)/y & y \neq 0, \\
g^\prime(0) & y = 0
\end{cases} 
\end{equation*}
is $C^1$ with absolutely continuous derivative $G^\prime$ of which the
 Radon-Nikodym derivative  $G^{\prime\prime}$ is
locally bounded.
\end{lemma}
\begin{proof} By \eqref{exp0},
 \begin{equation*}
  \frac{G(y) - G(0)}{y} = \frac{1}{2}g^{\prime\prime}(0) + O(y)
 \end{equation*}
and
\begin{equation*}
 G^\prime(y) = \frac{1}{y}\left(g^\prime(y)- \frac{g(y)}{y}\right)
= \frac{1}{2}g^{\prime\prime}(0) + O(y)
\end{equation*}
as $y\to 0$. In particular, $G$ is $C^1$. Further,
\begin{equation*}
 G^{\prime\prime}(y) = 
\frac{1}{y}\left(g^{\prime\prime}(y)- 
\frac{2}{y}\left(g^\prime(y)- \frac{g(y)}{y}\right)
\right) = O(1)
\end{equation*}
as $y\to 0$ 
by \eqref{exp0}, which means that $G^\prime$ is absolutely continuous
 and
$G^{\prime\prime}$ is locally bounded.
\end{proof}
\begin{lemma}\label{exist}
 The ordinary differential equation \eqref{ode} with
\eqref{gp}
has a unique $C^1$
 solution. The solution $g$ is $C^2$ with \eqref{exp0}
and satisfies $g(y)/y > 0$ for all
 $y \neq 0$.
\end{lemma}
\begin{proof}
Regarding \eqref{ode} as a quadratic equation in $g^\prime(y)$, we
get
\begin{equation*}
 g^\prime(y) = \frac{(1-2H)f(y) \pm \sqrt{(1-2H)^2f(y)^2 + 8Hq(y)}}{2q(y)},
\end{equation*}
where 
\begin{equation*}
f(y) = \frac{y}{g(y)}, \ \ 
q(y) =  1 + 2\rho \frac{y}{2H+1} + \frac{y^2}{(2H+1)^2}.
\end{equation*}
Note that $q(y) \geq 1-\rho^2 > 0$.
By \eqref{gp}, we should solve
\begin{equation*}
 g^\prime(y) = \varphi\left(
y, \frac{y}{g(y)}
\right),
\end{equation*}
where
\begin{equation}\label{varphi}
\varphi(y,z) = 
\frac{(1-2H)z + \sqrt{(1-2H)^2z^2 + 8Hq(y) }}{2q(y)}.
\end{equation}
Letting $y\to 0$, we should have
\begin{equation*}
 g^\prime(0) = \varphi\left(0,\frac{1}{g^\prime(0)}\right)
\end{equation*}
of which the unique solution is $g^\prime(0) = 1$.
The result then follows from Lemma~\ref{lemA3} below.
\end{proof}

\begin{lemma}\label{lemA3}
Let $\varphi : \mathbb{R}\times [0,\infty) \to (0,\infty)$ be a $C^2$
 function with
\begin{equation} \label{inc}
 \inf_{(y,z) \in \mathbb{R}\times [0,\infty)}  \frac{\partial
  \varphi}{\partial z}(y,z) \geq 0,
\end{equation}
\begin{equation} \label{Lip}
 \sup_{(y,z) \in \mathbb{R}\times [0,\infty)}  \frac{\partial
  \varphi}{\partial z}(y,z) < \infty,
\end{equation}
\begin{equation} \label{shr}
 \frac{\partial \varphi}{\partial z}(0,1) <  1, 
\end{equation}
\begin{equation}\label{zo}
 \varphi(0,1) = 1,
\end{equation}
and such that the equation
\begin{equation}\label{aleq}
 \alpha = \varphi\left(0,\frac{1}{\varphi(0,\frac{1}{\alpha})}\right),\
  \ \alpha > 0
\end{equation}
has the unique solution $\alpha = 1$.
Then, there exists a unique solution of the differential equation
\begin{equation*}
 g^\prime(y) = \varphi\left(
y, \frac{y}{g(y)}
\right), \ \ g(0) = 0
\end{equation*}
The solution $g$ is $C^2$ with \eqref{exp0} 
and satisfies $g(y)/y > 0$ for $y \neq 0$.
\end{lemma}
\begin{proof}
{\bf Step 1:} Here we show that 
there exists a $C^1$ function $g$ which satisfies
the integral equation $g = \Phi[g]$, where
\begin{equation*}
 \Phi[g](y) = \int_0^y \varphi\left(u,\frac{u}{g(u)}\right)\mathrm{d}u.
\end{equation*}
Since $\varphi(y,z)$ is increasing in $z$, we have
$\Phi[g_0](y)/y \geq \Phi[g_1](y)/y$ for all $y\neq 0$ 
if $0 < g_0(y)/y \leq g_1(y)/y$ for all $y \neq 0$.
Let 
\begin{equation*}
 g_0(y) = \int_0^y\varphi(u,0) \mathrm{d}u.
\end{equation*}
 and define  $g_{n+1} = \Phi[g_n]$ for $n \geq 0$.
By \eqref{inc} we have $0 < g_0(y)/y \leq g_1(y)/y$
for all $y \neq 0$. Then by the above mentioned monotonicity of $\Phi$, 
we have
$g_1(y)/y \geq g_2(y)/y$ for all $y \neq 0$.
Again by \eqref{inc} we have
$g_2(y)/y \geq g_0(y)/y$ for all $y \neq 0$.
This inequality then implies 
$g_3(y)/y \leq g_1(y)/y$ for all $y \neq 0$.
By induction we obtain
\begin{equation}\label{ind}
%0 <  \frac{y}{g_n(y)} \leq \frac{y}{g_0(y)}
%, \ \ g_n^\prime(y) \geq \sqrt{\frac{2H}{q(y)}}
0 < \frac{g_0(y)}{y} \leq \frac{g_{2n}(y)}{y} \leq \frac{g_{2n+2}(y)}{y}
\leq \frac{g_{2n+3}(y)}{y} \leq \frac{g_{2n+1}(y)}{y}
\end{equation}
for all $n \geq 0$ and for all $y \neq 0$. 
By the monotonicity of the sequences there exist $g_e(y)$ and
 $g_o(y)$ such that
\begin{equation}\label{sand}
\frac{g_{2n}(y)}{y} \leq
\lim_{n\to \infty} \frac{g_{2n}(y)}{y} =  \frac{g_e(y)}{y} \leq
 \frac{g_o(y)}{y} =  \lim_{n\to \infty} \frac{g_{2n + 1}(y)}{y} 
\leq  \frac{g_{2n + 1}(y)}{y} 
\end{equation}
for each $y \neq 0$. By the dominated convergence theorem, 
$(g_e,g_o)$ is a solution of the coupled integral equation
\begin{equation*}
  g_e = \Phi[g_o],   \ \ g_o = \Phi[g_e].
\end{equation*}
We are going to show $g_e = g_o$. Letting $y \to 0$ in \eqref{sand},
we have
\begin{equation}\label{sand2}
 \alpha_{2n} \leq \liminf_{y \to 0} \frac{g_e(y)}{y}  \leq
\limsup_{y \to 0} \frac{g_o(y)}{y}  \leq
\alpha_{2n+1},
\end{equation}
where,  by L'H\^opital's rule,
\begin{equation*}
 \alpha_n = \lim_{y\to0}
\frac{g_n(y)}{y} = \lim_{y \to 0}\varphi\left(
y,\frac{y}{g_{n-1}(y)}
\right) = \varphi\left(0,\frac{1}{\alpha_{n-1}}\right)
\end{equation*}
for $n \geq 1$ and $\alpha_0 = \varphi(0,0) > 0$.
By \eqref{ind}, both $\{\alpha_{2n}\}$ and
 $\{\alpha_{2n+1}\}$ are bounded 
monotone sequences and so convergent.
The limit $\alpha$ of each sequence has to be a solution of \eqref{aleq}
and so,  $\alpha = 1$ by the assumption.
We then conclude that $\alpha_n$ itself converges to $1$.
Now taking $n\to \infty$ in \eqref{sand2}, we have
\begin{equation}\label{lim1}
 \lim_{y \to 0} \frac{g_e(y)}{y}  
= \lim_{y \to 0} \frac{g_o(y)}{y}  
= 1.
\end{equation}
Therefore for any $\epsilon > 0$, there exists $\delta(\epsilon) > 0$ such that if $|y|\leq \delta(\epsilon)$,
\begin{equation}\label{delb}
1 + \epsilon \geq \frac{y}{g_o(y)}   \geq  \frac{y}{g_e(y)} \geq 1-\epsilon.
\end{equation}
On the other hand,
by \eqref{shr}, there exists $\delta_1 > 0$ such that
\begin{equation}\label{delb2}
 \sup_{|y|\leq \delta_1, |z-1| \leq \delta_1} \left|
 \frac{\partial \varphi}{\partial z}(y,z)\right| < 1- \frac{1}{2}\left(1-\frac{\partial
 \varphi }{\partial z}(0,1)\right).
\end{equation}
Let $\delta = \min\{\delta_1,\delta(\delta_1),\delta(\epsilon)\}$ and
\begin{equation*}
 \|f\|_\delta = \sup_{|y|\leq \delta} \left|\frac{f(y)}{y}\right|
\end{equation*}
for a function $f$. Then, using \eqref{delb} and \eqref{delb2},
\begin{equation*}
\begin{split}
\left\|
g_o-g_e
\right\|_\delta
&\leq \sup_{|y|\leq \delta}\frac{1}{y}\int_0^y
\left|
\varphi\left(u,\frac{u}{g_e(u)}\right) -
\varphi\left(u,\frac{u}{g_o(u)}\right)
\right|\mathrm{d}u
 \\
&\leq  \sup_{|y|\leq \delta_1, |z-1| \leq \delta_1} \left|
 \frac{\partial \varphi}{\partial z}(y,z)\right|
\sup_{|y|\leq \delta} \frac{1}{y}\int_0^y 
\frac{u^2}{g_e(u)g_o(u)}
\left|
\frac{g_e(u)}{u} - \frac{g_o(u)}{u}
\right|\mathrm{d}u \\
& \leq (1+\epsilon)^2
\left(1- \frac{1}{2}\left(1-\frac{\partial
 \varphi }{\partial z}(0,1)\right)\right)
  \|g_o-g_e\|_\delta.
\end{split}
\end{equation*}
We can take such $\epsilon > 0$  that
\begin{equation*}
 (1+\epsilon)^2
\left(1- \frac{1}{2}\left(1-\frac{\partial
 \varphi }{\partial z}(0,1)\right)\right) < 1
\end{equation*}
to conclude that $g_o(y) = g_e(y)$ for $|y|\leq
 \delta$.
For $y > \delta$, by the Lipschitz continuity of $\varphi$ in $z$,
 \eqref{sand} and \eqref{sand2}, we have
\begin{equation*}
| g_o(y) - g_e(y)| \leq L \int_{\delta}^y
\frac{y}{g_0(y)^2}|g_o(u)-g_e(u)|\mathrm{d}u
\end{equation*}
for a constant $L$. Then, by Gronwall's lemma we have
$g_o(y) = g_e(y)$ for $y \geq \delta$.
Similarly we obtain $g_o(y) = g_e(y)$ for $y \leq -\delta$.
Thus $g_o = g_e$ is a solution of $g = \Phi[g]$.
By \eqref{lim1}, the solution $g$ satisfies 
\begin{equation*}
 g^\prime(0) = 1 = \varphi(0,1) = \lim_{y \to 0} \varphi\left(
y,\frac{y}{g(y)}
\right) = \lim_{y \to 0} g^\prime(y)
\end{equation*}
and so, is a $C^1$ function.

{\bf Step 2:} Here we show that a solution of $g = \Phi[g]$ is unique.
Let $g$ and $\hat{g}$ be two solutions.
By \eqref{inc}, we have
\begin{equation*}
0 <  \frac{y}{g(y)} \leq \frac{y}{g_0(y)}, \ \ 
0 <  \frac{y}{\hat{g}(y)} \leq \frac{y}{g_0(y)}.
\end{equation*}
Therefore by the Lipschitz continuity of $\varphi$ in $z$,
if there exists $y_0 > 0$ such that $g(y_0) = \hat{g}(y_0)$, then
\begin{equation*}
| g(y) - \hat{g}(y)| \leq L \int_{y_0}^y
\frac{y}{g_0(y)^2}|g(u)-\hat{g}(u)|\mathrm{d}u
\end{equation*}
for $y \geq y_0$. We have $g(y) = \hat{g}(y)$ for $y\geq y_0$
by Gronwall's lemma.
Now, suppose that there exists $y >0$ such that $g(y) > \hat{g}(y)$.
Then, from the above observation, we have
 $g(u) > \hat{g}(u)$ for all $u\in
 (0,y)$.
However the monotonicity of $\varphi$ in $z$ implies that
\begin{equation*}
 0 < g(y) - \hat{g}(y)
= \int_0^y \varphi\left(u,\frac{u}{g(u)}\right) -
\varphi\left(u,\frac{u}{\hat{g}(u)}\right) \mathrm{d}u \leq 0
\end{equation*}
that is a contradiction. 
If there exists $y < 0$ such that $g(y) > \hat{g}(y)$,
then, again by a similar argument we conclude that
$g(u) > \hat{g}(u)$ for all $u \in (y,0)$.
This results in a contradiction as
\begin{equation*}
 0 < g(y)-\hat{g}(y) = -\int_y^0
\varphi\left(u,\frac{u}{g(u)}\right) -
\varphi\left(u,\frac{u}{\hat{g}(u)}\right)\mathrm{d}u < 0.
\end{equation*}
Therefore we have $g = \hat{g}$.

{\bf Step 3:} It remains to show that the solution $g$ is $C^2$ with \eqref{exp0}. Let
\begin{equation*}
 \hat{g}(y) = y + \beta \frac{y^2}{2},
\end{equation*}
where
\begin{equation*}
 \beta = \frac{\frac{\partial \varphi}{\partial y}(0,1)}{1 +\frac{1}{2} \frac{\partial \varphi}{\partial z}(0,1)}.
\end{equation*}
We have
\begin{equation}\label{expan}
 \frac{g(y)}{y} = \frac{\hat{g}(y)}{y} + O(y^2)
\end{equation}
as $y \to 0$. Indeed, since
\begin{equation}\label{expan2}
 \begin{split}
  \varphi\left(
y, \frac{y}{\hat{g}(y)}
\right) &= \varphi(0,1) + \frac{\partial \varphi}{\partial y}(0,1)y
+ \frac{\partial \varphi}{\partial z}(0,1)
\left(\frac{1}{1 + \beta y/2} -1 \right) + O(y^2) \\
&= 1 + \left(
\frac{\partial \varphi}{\partial y}(0,1)
-\frac{\beta}{2}\frac{\partial \varphi}{\partial z}(0,1)
\right)y + O(y^2) \\
&= 1 + \beta y + O(y^2) \\
& = \hat{g}^\prime(y) + O(y^2),
 \end{split}
\end{equation}
in light of \eqref{delb2}, there exists $L \in (0,1)$ such that
\begin{equation*}
 \begin{split}
  \sup_{0 < |y| \leq a} \left|
\frac{g(y)}{y} - \frac{\hat{g}(y)}{y}
\right| & \leq   \sup_{0 < |y| \leq a} \left|
\frac{1}{y}\int_0^y \varphi\left(
u,\frac{u}{g(u)}
\right) - \varphi\left(u,\frac{u}{\hat{g}(u)}\right) \mathrm{d}u \right|
  + O(a^2)\\
& \leq L
  \sup_{0 < |y| \leq a} \left|
\frac{g(y)}{y} - \frac{\hat{g}(y)}{y}
\right| + O(a^2).
 \end{split}
\end{equation*}
This implies
\begin{equation*}
  \sup_{0 < |y| \leq a} \left|
\frac{g(y)}{y} - \frac{\hat{g}(y)}{y}
\right| = O(a^2)
\end{equation*}
and in particular, \eqref{expan}.
Now, from \eqref{expan}, we have
\begin{equation}\label{gpp}
\begin{split}
 \frac{1}{y}\left(\frac{g(y)}{y} -1\right) & = 
 \frac{1}{y}\left(\frac{\hat{g}(y)}{y} -1\right) + O(y) \\
& = \frac{\beta}{2} + O(y) 
\end{split}
\end{equation}
as $y \to 0$. Further, by \eqref{expan} and \eqref{expan2},
\begin{equation}\label{gpp2}
\begin{split}
%\lim_{y\to 0} \frac{g^\prime(y) -1}{y} &=  
%\lim_{y\to 0}
% \frac{1}{y}\left(\varphi\left(y,\frac{y}{g(y)}\right)-1\right)\\
% & =  \frac{\partial \varphi}{\partial y}(0,1)
%+ 
% \frac{\partial \varphi}{\partial z} (0,1)
%\lim_{y\to 0} \frac{1}{y}\left(\frac{y}{g(y)}-1\right) \\
%& = \frac{\partial \varphi}{\partial y}(0,1)
%-\frac{\beta}{2}
% \frac{\partial \varphi}{\partial z} (0,1)\\
%&= \beta.
\frac{g^\prime(y) -1}{y} = \frac{\hat{g}^\prime(y) -1}{y} + O(y)
= \beta + O(y)
\end{split}
\end{equation}
as $y\to 0$.
On the other hand, 
\begin{equation}\label{gpp3}
 \begin{split}
  g^{\prime\prime}(y) &= 
 \frac{\partial \varphi}{\partial y}\left(y,\frac{y}{g(y)}\right) +
 \frac{\partial \varphi}{\partial z}\left(y,\frac{y}{g(y)}\right)
\left(\frac{y}{g(y)}\right)^\prime \\
&=
\frac{\partial \varphi}{\partial y}\left(y,\frac{y}{g(y)}\right) +
 \frac{\partial \varphi}{\partial z}\left(y,\frac{y}{g(y)}\right)
\frac{1}{g(y)}\left(1 - \frac{y}{g(y)}- 
\frac{y}{g(y)}(g^\prime(y)-1)\right)
\\
& =
\frac{\partial \varphi}{\partial y}\left(0,1\right) +
 \frac{\partial \varphi}{\partial z}\left(0,1\right)
\left(\frac{\beta}{2}-\beta\right) + O(y) \\
&= \beta + O(y)
 \end{split}
\end{equation}
as $y\to 0$.
Therefore, $g$ is $C^2$ with \eqref{exp0}.
\end{proof}

\end{document}